\documentclass[12pt]{article}
\usepackage{amsfonts}
\usepackage{amsmath}
\usepackage{amssymb}
\usepackage{xcolor}

\newtheorem{theorem}{Theorem}

\newtheorem{definition}[theorem]{Definition}

\newtheorem{lemma}[theorem]{Lemma}

\newenvironment{proof}[1][Proof]{\noindent\textbf{#1.} }{\ \rule{0.5em}{0.5em}}
\begin{document}

\title{Uniformly Strict Equilibrium for Repeated Games with Private
Monitoring and Communication\thanks{%
We thank George Mailath for very helpful discussions and the audience in
numerous conferences and seminars for helpful comments. Postlewaite
gratefully acknowledges support from National Science Foundation grant
SBR-9810693.}}
\author{Richard McLean \\
Rutgers University \and Ichiro Obara \\
UCLA \\
\and Andrew Postlewaite \\
University of Pennsylvania}
\date{February 5, 2024}
\maketitle

\begin{abstract}
Cooperation through repetition is an important theme in game theory. In this
regard, various celebrated ``folk theorems'' have been proposed for repeated
games in increasingly more complex environments. There has, however, been
insufficient attention paid to the robustness of a large set of equilibria
that is needed for such folk theorems. Starting with perfect public
equilibrium as our starting point, we study uniformly strict equilibria in
repeated games with private monitoring and direct communication (cheap
talk). We characterize the limit equilibrium payoff set and identify the
conditions for the folk theorem to hold with uniformly strict equilibrium.


JEL\ Classifications:\ C72, C73, D82
\end{abstract}

\pagebreak

\section{Introduction}

Cooperation through repetition is an important theme in game theory. In this
regard, various celebrated ``folk theorems'' have been proposed for repeated
games in increasingly more complex environments. There has, however, been
insufficient attention paid to the robustness of a large set of equilibria
that is needed for such folk theorems.

In this paper, we study \emph{uniformly strict equilibria} in repeated games
with private monitoring and direct communication (cheap talk). Our starting
point is perfect public equilibrium (PPE) (\cite{flm}). In each period,
players take actions simultaneously, observe private signals, and send
public messages simultaneously. A perfect public equilibrium is a profile of
public strategies that specifies a Nash equilibrium as their continuation
play after every public history (a sequence of past message profiles). We
impose strict incentives at every public history by requiring that, in each
period, a player would incur a positive payoff loss (in terms of the value
at the period) when deviating in either action or message from the
equilibrium strategy. We also require such payoff losses from a unilateral
deviation to be uniformly bounded away from $0$ across all public histories.

It is well known that strict equilibrium has desirable robustness
properties. For example, strict equilibria survive most equilibrium
refinements in strategic form games. In our setting of infinitely repeated
games, our uniform strictness requirement is a natural strengthening of
strict equilibrium.

We present two main results. Our first result is a characterization of the
limit set of uniformly strict perfect public equilibrium payoffs via a
collection of static programming problems. We follow the approach of
Fudenberg and Levine \cite{fl} (henceforth FL) to characterize the limit
equilibrium payoff set. It also builds on other classic results from Abreu,
Pearce and Stacchetti \cite{aps} and Fudenberg, Levine and Maskin \cite{flm}%
. We adapt their ideas to our model and generalize them by introducing
uniformly strict incentives. In our second result, we establish a folk
theorem by identifying conditions ensuring that this limit set coincides
with the set of feasible and individually rational payoffs generated by the
data of the underlying stage game.

There is a large literature dealing with folk theorems for repeated games
with varying assumptions regarding public or private monitoring with or
without communication. Most relevant to our paper are the various folk
theorems for repeated games with private monitoring and communication (\cite%
{ao}, \cite{bpk}, \cite{c}, \cite{fl2}, \cite{km}, \cite{l}, \cite{l2}, \cite%
{o}, \cite{t}).\footnote{%
There is an extensive literature on folk theorems for repeated games with
private monitoring and without communication, including \cite{bo}, \cite{ev}
, \cite{ho}, \cite{mm1}, \cite{mo}, \cite{m}, \cite{p}, \cite{s}, \cite{su}.
They usually rely on non-strict equilibrium (such as belief-free
equilibrium) to establish the folk theorem.}

Our detectability and identifiability conditions for the folk theorem are similar to and weaker
than the conditions (A2) and (A3) in Kandori and Matsushima \cite{km}. (A2)
and (A3) imply that, for any pair of players, their deviations are
detectable and identifiable (i.e. one player's deviation can be
statistically distinguished by the other player's deviation) based on the
private signals of the other $n-2$ players. Our detectability and
identifiability condition instead impose a similar restriction on the joint
distributions of the messages of all players. Their conditions allow each player's future payoff independent of her message. This indifference makes truth-telling incentive compatible for each player. On the other hand, we require uniform strictness of incentive for sending any (nontrivial) message.\footnote{\cite{km}
also discusses a way to provide strict incentive for truth-telling via a
scoring rule, but the strict incentive vanishes in the limit for the minmax points for their folk
theorem (Theorem 2). We instead fix the level of the required strict incentive
first, then prove the folk theorem by letting $\delta \rightarrow 0$.}

Our conditions are also similar to the sufficient
conditions in Tomala \cite{t}, but Tomala studies a type of perfect
equilibrium with mediated communication, which is more flexible than cheap
talk, and does not impose strict incentive constraints. As a consequence,
the conditions for the folk theorem in Tomala \cite{t} are weaker than ours.

\section{Preliminaries}

\subsection{Repeated Games with Private Monitoring and Communication}

\textbf{Stage Game}

\vspace{3mm}

We present the model of repeated games with private monitoring and
communication. The set of players is $N=\left\{ 1,...,n\right\} $. The game
proceeds in stages and in each stage $t$, player $i$ chooses an action from
a finite set $A_{i}$. \ An action profile is denoted by $a=\left(
a_{1},...,a_{n}\right) \in \Pi _{i}A_{i}:=A.$ Stage game payoffs are given
by $g:A \rightarrow \mathbb{R}^n$. We denote the resulting stage game by $%
G=\left(N, A, g\right) .$ Actions are not publicly observable. Instead, each
player $i$ observes a private signal $s_{i}$ from a finite set $S_{i}$. A
private signal profile is denoted $s=\left( s_{1},..,s_{n}\right) \in \Pi
_{i}S_{i}:=S.$ For each $a\in A,$ $p\left( \cdot |a\right) \in \Delta (S)$
is the distribution on $S$ given action profile $a.$ We assume that the
marginal distributions have full support, that is, $\sum_{s_{-i}}p\left(
s_{i},s_{-i}|a\right) >0$ for all $s_{i}\in S_{i}$, $a\in A$ and $i\in N$.

Players communicate publicly each period. Player $i$ sends a public message $%
m_{i}$ from a finite set $M_{i}$ after observing a private signal $s_{i}$ in
each period.\footnote{%
We can support the largest set of equilibria by using $M_{i}=S_{i}$. But we
use a more general message space $M_{i}$ to allow for the possibility of
restricted message spaces.} Player $i$'s \emph{message strategy} $\rho
_{i}:S_{i}\rightarrow M_{i}$ in the stage game is a mapping from private
signals to public messages. Let $R_{i}$ be the set of player $i$'s message
strategies. An action profile $a\in A$ and a profile of message strategies $%
\rho =(\rho _{1},...,\rho _{n})\in \Pi _{i}R_{i}:=R$ generates a
distribution $\tilde{p}(\cdot |(a,\rho ))$ over public messages $M=\Pi
_{i}M_{i}$, where 
\begin{equation*}
\tilde{p}(m|(a,\rho )):=\sum_{\substack{ s\in S:\rho
_{i}(s_{i})=m_{i},\forall i}}p(s|a).
\end{equation*}

We normalize payoffs so that each player's \textit{pure strategy} minmax
payoff is $0$ in the stage game. The pure strategy minmax payoff is the
relevant payoff lower bound for our folk theorem because we study
equilibrium with strict incentives and without any mediator. Note that the
pure strategy minmax payoff may be strictly larger than the mixed minmax
payoff. The set of feasible payoff profiles is $V(G)=co\left\{ g\left(
a\right) |a\in A\right\} $. Let $A(G)\subseteq A$ be the set of action
profiles that generate an extreme point in $V(G)$. Finally $V^{\ast
}(G)=\left\{ v\in V(G)|v\geq 0\right\} $ is the set of feasible,
individually rational payoff profiles.

\vspace{3mm}

\textbf{Repeated Game with Public Communication}

\vspace{3mm}

In the repeated game, play proceeds in the following way. At the beginning
of period $t\geq 1$, player $i$ chooses an action contingent on $\left(
h_{i}^{t},h^{t}\right)$, where $h_{i}^{t}\in H_{i}^{t}=A_{i}^{t-1}\times
S_{i}^{t-1}$ is player $i$'s private history that consists of her private
actions and private signals and $h^{t}\in H^{t}=M^{t-1}$ is the public
history of message profiles.\footnote{%
We define $H^{1}=H^{1}_i = \{\varnothing \}$ for all $i \in N$.} Player $i$
also chooses a message strategy $\rho _{i} \in R_i$ contingent on $\left(
h_{i}^{t},h^{t},a_{i}\right)$. Then player $i$'s pure strategy $\sigma
_{i}=(\sigma _{i}^{a},\sigma _{i}^{m})$ consists of an \textquotedblleft
action\textquotedblright\ component $\sigma
_{i}^{a}:\bigcup_{t}[H_{i}^{t}\times H^{t}]\longrightarrow A_{i}$ and a
\textquotedblleft message\textquotedblright\ component $\sigma
_{i}^{m}:\bigcup_{t}[H_{i}^{t}\times H^{t}\times A_{i}]\longrightarrow R_{i}$%
.

A strategy $\sigma _{i}$ is a public strategy if in any period $t$ both $%
\sigma _{i}^{a}(h_{i}^{t},h^{t})$ and $\sigma
_{i}^{m}(h_{i}^{t},h^{t},a_{i}) $ are independent of private history 
$h_i^{t}\in H_i^{t}$. For the sake of simple exposition, we drop $h_{i}^{t}$ from any public strategies. 
We denote player $i$'s action and on-path message strategy at $h^t$ for public strategy $\sigma_i$ by $\sigma_i(h^t) = (\sigma _{i}^{a}(h^{t}), \sigma _{i}^{m}(h^{t}, \sigma _{i}^{a}(h^{t}))) \in A _i \times R_i$.
A pure strategy profile $\sigma $ induces a probability measure on $%
A^{\infty } $. Player $i$'s discounted average payoff given a profile of
strategies $\sigma =(\sigma _{1},...,\sigma _{n})$ is $(1-\delta
)E[\sum_{t=1}^{\infty }\delta ^{t-1}g_{i}(\tilde{a}^{t})|\sigma ]$, where
the expectation is taken with respect to this measure.

\subsection{Uniformly Strict Perfect Public Equilibrium}

A profile of public strategies $\sigma $ is a perfect public equilibrium if
its continuation strategies constitute a Nash equilibrium after every public
history(\cite{flm}). In this paper, we impose an additional robustness
requirement by requiring uniformly strict incentive compatibility at every
public history. Let $w^{\sigma }(h^{t})$ be a profile of discounted average
continuation payoffs at public history $h^{t}\in H$ given a public strategy
profile $\sigma $. Given $\sigma $ and $h^{t}$, let $\Sigma _{i}^{\sigma
,h^{t}}$ be the set of deviations $(a_{i}^{\prime },\rho _{i}^{\prime })\in
A_{i}\times R_{i}$ such that $a_{i}^{\prime }\neq \sigma _{i}^{a}(h^{t})$ or 
$i$'s unilateral deviation from $\sigma _{i}(h^{t})$ to $(a_{i}^{\prime
},\rho _{i}^{\prime })$ changes the distribution of continuation payoff
profiles $w^{\sigma }(h^{t},\cdot )$ from period $t+1$. We call such
one-shot deviations \emph{nontrivial deviations} at $h^{t}$ with respect to $%
w^{\sigma }$. Any other one-shot deviation is called a \emph{trivial
deviation}, as it does not change any outcome in the current period and in
the future at all.

We define \emph{$\eta$-uniformly strict perfect public equilibrium} ($\eta$%
-USPPE) as follows.

\begin{definition}
{($\eta$-USPPE)} A profile of public strategies $\sigma$ is an $\eta$
-uniformly strict perfect public equilibrium for $\eta \geq 0$ if the
following conditions are satisfied for any $h^{t} \in H$ for any $%
(a_{i}^{\prime },\rho _{i}^{\prime }) \in \Sigma_i^{\sigma, h^t}$ and any $i
\in N$, \newline
\begin{align*}
& g_{i}\left( \sigma ^{a}\left( h^{t}\right) \right) +\frac{\delta }{%
1-\delta }\sum_{m\in M}w_{i}^{\sigma }\left( h^{t},m\right) \tilde{p}\left(
m|\sigma \left( h^{t}\right) \right) -\eta \geq \\
& g_{i}\left( a_{i}^{\prime },\sigma _{-i}^{a}\left( h^{t}\right) \right) +%
\frac{\delta }{1-\delta }\sum_{m\in M}w_{i}^{\sigma }(h^{t},m)\tilde{p}%
\left( m |(a_{i}^{\prime },\rho _{i}^{\prime }),\sigma _{-i}\left(
h^{t}\right) \right),
\end{align*}
\end{definition}

This condition means that player $i$ would lose at least $\eta$ at any
public history if she makes any nontrivial deviation.\footnote{%
The incentive constraints for trivial deviations are satisfied by definition.%
}$^{,\thinspace}$\footnote{%
Another possible formulation of uniformly strict equilibrium would be to
require such $\eta$-strict incentive uniformly across all the information
sets, including the interim stages after observing a private signal.}

This definition just checks the one-shot deviation
constraints at each public history, but all the incentive constraints for
the continuation game after each public history are satisfied, because the
one-shot deviation principle holds.

Note that $\eta$-USPPE $\sigma$ may assign a suboptimal message off-path, i.e. $\sigma_i^m(h^t,a^\prime_i)$ may not be an optimal message strategy when $a_i^\prime \neq \sigma_i^a(h^t)$, since it is just a Nash equilibrium. But we can replace them with an optimal message to obtain a sequential equilibrium that is realization equivalent to $\sigma$, because other players never learn about player $i$'s deviation to $a_i^\prime$ due to the full support assumption.\footnote{Note that we do not require any strict incentive for such off-path message strategies and any trivial deviation from the on-path messages, as they do not affect any player's incentive or payoff at all.}

As an example of $\eta $-USPPE, consider any stage game with an $\eta $%
-strict Nash equilibrium. Then repeating this $\eta $-strict Nash
equilibrium and sending some message independent of histories is an $\eta $%
-uniformly strict PPE.\footnote{%
Note that any deviation in message after playing the Nash equilibrium is a
trivial deviation for this strategy profile.} In the following, let $E^{\eta
}(\delta )\subset \mathbb{R}^{n}$ denote the set of all $\eta $-USPPE payoff
profiles given $\delta $. In general, $\eta $-USPPE may not exist, hence $%
E^{\eta }(\delta )$ may be an empty set. 
The equilibrium payoff set for the standard PPE is compact, but the compactness of $
E^{\eta }(\delta )$ may not be entirely obvious because the set of nontrivial deviations at each public history depends on the continuation payoff profile. However, we can show that $E^\eta(\delta)$ is compact.

\begin{lemma}
\label{compact} $E^\eta(\delta)$ is compact.
\end{lemma}

\begin{proof}
See the Appendix.
\end{proof}

\section{Characterization of Limit Equilibrium Payoff Set}

\subsection{Constructing The Bounding Set for Equilibrium Payoffs}

We characterize the limit $\eta $-USPPE payoff set in two steps. In this
subsection, we construct a compact set $Q^{\eta }$ with the property that $%
E^{\eta }(\delta )\subseteq Q^{\eta }$ for all $\delta \in (0,1).$ In the
next subsection, we show that, if $intQ^{\eta }\neq \emptyset $, then for
any $\epsilon >0$, there exists a nonempty, compact, convex set $W\subseteq
Q^{\eta }$ and $\underline{\delta }\in \left( 0,1\right) $ such that $%
W\subseteq E^{\eta }(\delta )$ for any $\delta \in \left( \underline{ \delta 
},1\right) $ and the Hausdorff distance between $W$ and $Q^{\eta }$ is less
than $\epsilon $.

Let $\Lambda =\left\{ \lambda \in \mathbb{R}^{n}|\left\Vert
\lambda\right\Vert =1\right\} $ and $e^{i}=(0,0,..,1,...,0)^{\top }\in
\Lambda$ with the $i$th coordinate equal to 1. Following the approach of
Fudenberg and Levine \cite{fl}, for each $\lambda \in \Lambda $, we consider
the following programming problem $(P^{\lambda, \eta})$.

\begin{align*}
(P^{\lambda ,\eta })\ & \sup_{v\in \mathbb{R}^{n},\ a\in
A,\ \rho \in R,\ x:M\rightarrow \mathbb{R}^{n}}\lambda \cdot v\ s.t. \\
& v=g\left( a\right) +E[x\left( \cdot \right) |(a,\rho )] \\
& g_{i}\left( a\right) +E[x_{i}\left( \cdot \right) |(a,\rho )]-\eta \geq
g_{i}\left( a_{i}^{\prime },a_{-i}\right) +E[x_{i}\left( \cdot \right)
|(a_{i}^{\prime },\rho _{i}^{\prime }),(a_{-i},\rho _{-i})] \\
& \forall (a_{i}^{\prime },\rho _{i}^{\prime })\in \hat{\Sigma}_{i}^{(a,\rho
),x}\ \forall i\in N \\
& \sum_{i}\lambda _{i}x_{i}(m)\leq 0\ \forall m\in M
\end{align*}
where $\hat{\Sigma}_{i}^{(a,\rho ),x}$ is the set of $(a_{i}^{\prime },\rho
_{i}^{\prime })\in A_{i}\times R_{i}$ such that $a_{i}^{\prime }\neq a_{i}$
or player $i$'s unilateral deviation from $(a_{i},\rho _{i})$ to $%
(a_{i}^{\prime },\rho _{i}^{\prime })$ does not change the distribution of $%
x(\cdot )$ given $(a_{-i},\rho _{-i})$. Naturally, we call such a deviation
nontrivial deviation with respect to $x$ for $(P^{\lambda ,\eta })$.

Since the value of the problem is bounded above by $\max_{a \in A} \lambda
\cdot g(a)$, it is either some finite value or $- \infty$ when there is no
feasible solution.

This programming problem is different from FL's problem in \cite{fl} in two
aspects because of our uniform strictness requirement. First, an $\eta$
-wedge is added to the incentive constraints for nontrivial deviations (with
respect to $x(\cdot )$). Secondly, we restrict attention to pure actions
because uniformly strict equilibrium must be in pure strategies by
definition. Note that this problem is independent of $\delta $ like FL's
problem.

Let $k^{\eta }(\lambda )$ denote the vale of the supremum for $(P^{\lambda
,\eta })$. Let $H^{\eta }(\lambda )=\left\{ x\in \mathbb{R}^{n}|\lambda
\cdot x\leq k^{\eta }(\lambda )\right\} $ be the half space below the
hyperplane $\lambda \cdot x=k^{\eta }(\lambda )$ if $k^{\eta }(\lambda )$ is
finite. $H^{\eta }(\lambda )$ is an empty set if $k^{\eta }(\lambda
)=-\infty $. Let $Q^{\eta }=\bigcap_{\lambda \in \Lambda }H^{\eta }(\lambda )
$. The next theorem shows that $Q^{\eta }$ is a bound of the equilibrium
payoff set given any $\eta $ and $\delta $.

\begin{theorem}
\label{bound2} For any $\eta \geq 0$ and any $\delta \in (0,1),$ $E^{\eta
}(\delta) \subseteq Q^{\eta }$.
\end{theorem}

\begin{proof}
If $E^{\eta }(\delta )=\varnothing$, then $E^{\eta }(\delta )\subseteq
Q^{\eta }$ is trivially true. So suppose that $E^{\eta }(\delta )\neq
\varnothing$ and recall that $E^{\eta }(\delta )$ is a nonempty compact set
by Lemma \ref{compact}. Fix any $\eta \geq 0$ and pick any $\lambda \in
\Lambda $. Let $v^{\ast }$ be the $\eta $-uniformly strict PPE payoff
profile that solves $\max_{v\in E^{\eta }(\delta )}\lambda \cdot v$. Let $%
\sigma ^{\ast }$ be the equilibrium strategy profile to achieve $v^{\ast }$
and $(a^{\ast },\rho ^{\ast })\in A\times R$ be the equilibrium action
profile and the message strategy profile in the first period. Since $\sigma
^{\ast }$ is an $\eta $-USPPE, it must satisfy the following conditions: 
\begin{align*}
& g_{i}\left( a^{\ast }\right) +\frac{\delta }{1-\delta }\sum_{m\in
M}w_{i}^{\sigma ^{\ast }}\left( m\right) \tilde{p}\left( m|(a^{\ast },\rho
^{\ast })\right) -\eta \geq \\
& g_{i}\left( a_{i}^{\prime },a_{-i}^{\ast }\right) +\frac{\delta }{1-\delta 
}\sum_{m\in M}w_{i}^{\sigma ^{\ast }}(m)\tilde{p}(m|(a_{i}^{\prime },\rho
_{i}^{\prime }),(a_{-i}^{\ast },\rho _{-i}^{\ast }))\ \forall (a_{i}^{\prime
},\rho _{i}^{\prime })\in \Sigma _{i}^{\sigma ^{\ast },h^{1}}.
\end{align*}

Define $x_{i}^{\ast }(m)=\frac{\delta }{1-\delta }\left( w_{i}^{\sigma
^{\ast }}\left( m\right) -v_{i}^{\ast }\right) $. Then $\sum_{m}\lambda
_{i}x_{i}^{\ast }(m)\leq 0$ because $w_{i}^{\sigma ^{\ast }}\left( m\right)
\in E^{\eta }(\delta )$. Since $x_{i}^{\ast }$ is a translation of $\frac{%
\delta }{1-\delta }w_{i}^{\sigma ^{\ast }}$ by a constant, $(v^{\ast
},(a^{\ast },\rho ^{\ast }),x^{\ast })$ satisfies all the $\eta $-strict
incentive compatibility constraints with respect to the set of nontrivial
deviations $\hat{\Sigma}_{i}^{\sigma ^{\ast }(h^{1}),x^{\ast }}$ for player $%
i$ in the programming problem $(P^{\lambda ,\eta })$. Finally, 
\begin{align*}
& g\left( a^{\ast }\right) +E[x^{\ast }(\cdot )|(a^{\ast },\rho ^{\ast })] \\
& =g\left( a^{\ast }\right) +\frac{\delta }{1-\delta }\sum_{m\in M}\left(
w^{\sigma ^{\ast }}\left( m\right) -v^{\ast }\right) \tilde{p}\left(
m|(a^{\ast },\rho ^{\ast })\right) \\
& =v^{\ast }
\end{align*}%
and it follows that $(v^{\ast },(a^{\ast },\rho ^{\ast }),x^{\ast })$ is
feasible for $(P^{\lambda ,\eta })$ and $k^{\eta}(\lambda )$ is finite. This
implies $k^{\eta }(\lambda )\geq \lambda \cdot v^{\ast }$, hence $E^{\eta
}(\delta )$ is contained in the halfspace $H^{\eta }(\lambda )$. Since this
is true for all $\lambda \in \Lambda $, we have $E^{\eta }(\delta )\subset
\bigcap_{\lambda \in \Lambda }H^{\eta }(\lambda )=Q^{\eta }$ for any $\delta
\in (0,1)$.
\end{proof}

\vspace{3mm}

The following lemma, which corresponds to Lemma 3.2. in FL, is useful to
assess the possibility of a uniformly strict folk theorem.

\begin{lemma}
\label{bound} $k^{\eta }(-e^{i})$ is bounded above by $-\underline{v}_i^\eta$%
, where 
\begin{equation*}
\underline{v}_i^\eta = \min_{a \in A} \left[\max \left\{g_i(a),
\max_{a_{i}^{\prime }\neq a_{i}}g_{i}(a)+\eta\right\}\right].
\end{equation*}
\end{lemma}

\begin{proof}
Suppose that $(v,a,\rho ,x)$ is feasible for problem $(P^{-e^{i} ,\eta })$.
The last constraint of the problem becomes $x_{i}(m)\geq 0\ \forall m\in M$.
Then player $i$'s payoff 
\begin{equation*}
v_{i}=g_{i}(a)+\sum_{m\in M}x_{i}(m)\tilde{p}(m|(a,\rho ))
\end{equation*}
is bounded from below by $g_{i}(a)$. By the $\eta$-strict incentive
constraint, $v_{i}$ is also bounded from below by $\max_{a_{i}^{\prime }\neq
a_{i}}g_{i}(a_{i}^{\prime },a_{-i})+\eta.$ Hence $v_{i}$ is bounded below by 
\begin{equation*}
\underline{v}_i^\eta =\min_{a\in A}\left[ \max \left\{
g_{i}(a),\max_{a_{i}^{\prime }\neq a_{i}}g_{i}(a_{i}^{\prime },a_{-i})+\eta
\right\} \right]
\end{equation*}
Therefore, $k^{\eta }(-e^{i})$ is bounded from above by $-\underline{v}%
_i^\eta$.
\end{proof}

\bigskip

This $\underline{v}_i^\eta$ coincides with the minmax payoff $0$ when $\eta
=0$, but can be strictly positive when $\eta >0$. As the next lemma shows,
it coincides with the minmax payoff if and only if there exists a minmax
action profile for player $i$ where player $i$ plays an $\eta $-strictly
optimal action.

\begin{lemma}
\label{minmax} For each $i\in N$, $\underline{v}_{i}^{\eta }=0$ if there
exists $\underline{a}^{i}\in A$ such that $g_{i}(\underline{a}%
^{i})=\min_{a_{-i}^{\prime }}\max_{a_{i}^{\prime }}g_{i}(a_{i}^{\prime
},a_{-i}^{\prime })$ and $g_{i}(\underline{a}^{i})-g_{i}(a_{i}^{\prime },%
\underline{a}_{-i}^{i})\geq \eta $ for any $a_{i}^{\prime }\neq \underline{a}%
_{i}^{i}$. Furthermore, $\underline{v}_{i}^{\eta }>0$ if there is no such $%
\underline{a}^{i}\in A$.
\end{lemma}

\begin{proof}
Fix $i$ and choose any $a \in A$. If $a_i$ is a best response to $a_{-i}$,
then 
\begin{align*}
g_{i}(a)=\max_{a_{i}^{\prime}}g_{i}(a_{i}^{\prime},a_{-i}) \geq \min
_{a_{-i}^\prime}\max_{a_{i}^{\prime}}g_{i}(a_{i}^{\prime},a_{-i}^\prime) = 0.
\end{align*}
If not, then 
\begin{align*}
\max_{a_{i}^{\prime } \neq a_{i}}g_{i}(a_{i}^{\prime}, a_{-i}) +\eta \geq
\min
_{a_{-i}^\prime}\max_{a_{i}^{\prime}}g_{i}(a_{i}^{\prime},a_{-i}^\prime)
+\eta \geq 0.
\end{align*}
Hence, $\max \left\{ g_{i}(a),\max_{a_{i}^{\prime }\neq
a_{i}}g_{i}(a_{i}^{\prime },a_{-i})+\eta \right\}$ is nonnegative for any $a
\in A$.

Suppose that $\underline{a}^{i}\in A$ satisfies the conditions of the Lemma.
Then $g_{i}(\underline{a}^{i}) =0$ and $\max_{a_{i}^{\prime }\neq \underline{%
a}^{i}_{i}}g_{i}(a_{i}^{\prime },\underline{a}^{i}_{-i})+\eta \leq g_{i}(%
\underline{a}^{i}) =0$. Hence $\underline{v}_i^\eta = 0$ is achieved at $a = 
\underline{a}^{i}$.

Suppose that there is no such $\underline{a}^{i}\in A$. Then $\eta$ must be
strictly positive since the minmax action profile would satisfy the
conditions when $\eta = 0$. Take any $a\in A$. If $a_{i}$ is a best response
to $a_{-i}$, then $g_{i}(a)\geq \min_{a_{-i}^{\prime }}\max_{a_{i}^{\prime
}}g_{i}(a_{i}^{\prime },a_{-i}^{\prime })=0$. Hence we have either (1) $%
g_{i}(a)>0$ or (2) $g_{i}(a)=0$ but $g_{i}(a)-\max_{a_{i}^{\prime }\neq
a_{i}}g_{i}(a_{i}^{\prime },a_{-i})<\eta $, which implies $%
\max_{a_{i}^{\prime }\neq a_{i}}g_{i}(a_{i}^{\prime },a_{-i})+\eta >0$.
Hence, $\max \left\{ g_{i}(a),\max_{a_{i}^{\prime }\neq
a_{i}}g_{i}(a_{i}^{\prime },a_{-i})+\eta \right\} $ is strictly positive in
either case.

When $a_{i}$ is not a best response to $a_{-i}$, then 
\begin{equation*}
\max_{a_{i}^{\prime }\neq a_{i}}g_{i}(a_{i}^{\prime },a_{-i})+\eta \geq
\min_{a_{-i}^{\prime }}\max_{a_{i}^{\prime }}g_{i}(a_{i}^{\prime
},a_{-i}^{\prime }) + \eta > 0 
\end{equation*}
as $\eta > 0$. Since $A$ is a finite set, $\underline{v}_{i}^{\eta }=\min_{a
\in A} \left[ \max \left\{ g_{i}(a),\max_{a_{i}^{\prime }\neq
a_{i}}g_{i}(a_{i}^{\prime },a_{-i})+\eta \right\} \right] $ must be strictly
positive.
\end{proof}

\bigskip

If this bound $\underline{v}_i^\eta$ is strictly positive, then $k(-e^i) = - 
\underline{v}_i^\eta < 0$, hence the minmax payoff can never be approximated
by any $\eta$-USPPE by Theorem \ref{bound2}. So, it is necessary for a folk
theorem that an $\eta$-strict incentive is provided by the current payoffs
at the minmax point.

Also note that this bound may be achieved by some non-minmax action profile
when it is strictly positive. If no minmax action profile for player $i$ is $%
\eta $-strictly optimal for $i$, then some non-minmax action profile $\hat{a}%
\in A$ may achieve $\underline{v}_i^\eta > 0$ if $g_{i}(\hat{a})$ is close
to $0$ and $\max_{a_{i}^{\prime }\neq \hat{a} _{i}}g_{i}(a_i^\prime, \hat{a}%
_{-i})+\eta $ is small as well (any deviation from $\hat{a}$ is very costly
for player $i$).

Similarly, we observe that $k^{\eta }(e^{i})$ may be strictly below $%
\max_{a}g_{i}(a)$ unless it is $\eta$-strictly optimal for player $i$ to
play the action that achieves this value. Otherwise, an additional incentive
needs to be provided for player $i$ through some punishment (as $\lambda
=e^{i}$). This necessarily leads to some inefficiency because punishment
occurs with positive probability (Green and Porter \cite{gp}). Thus it is
necessary for a folk theorem that $g_{i}(a)-g_{i}(a_{i}^{\prime
},a_{-i})\geq \eta $ holds for every $a_{i}^{\prime }\neq a_{i}$ for some
action profile $a$ that solves $\max_{a}g_{i}(a)$. If this is not the case,
then $k^{\eta }(e^{i})$ must be less than $\max_{a}g_{i}(a)$ and may be
achieved by some action profile that does not solve $\max_{a}g_{i}(a)$.

\subsection{Limit Result for Equilibrium Payoff Set}

\subsubsection{Decomposability and Local Decomposability}

Our main theorem claims that $Q^{\eta }$ provides the limit $\eta $-USPPE
payoff set when $Q^{\eta }$ has an interior point. We prove it by
establishing $\eta $-uniformly strict versions of many well-known results in 
\cite{fl} and in \cite{aps}.

We first observe that a set of payoffs can be supported by $\eta $-USPPE if
it is self-decomposable with respect to $\eta $-strict incentive constraints
with respect to nontrivial deviations. Given $\delta \in (0,1)$ and $%
w:M\rightarrow \mathbb{R}^{n}$, we consider the static game $\Gamma ^{\delta
}(G,p,w)$, where player $i$'s strategy set is $A_{i}\times R_{i}$ and player 
$i$'s payoff is $\left( 1-\delta \right) g_{i}\left( a\right) +\delta
E[w_{i}(\cdot )|(a,\rho)]$, where $w$ assigns payoffs for each message
profile and the expectation is computed with respect to $\tilde{p}(\cdot
|a,\rho )$.

\begin{definition}
A pair consisting of an action profile $a\in A$ and a profile of message
strategies $\rho \in R$ is $\eta $-\emph{enforceable} for $\eta >0$ with
respect to nonempty set $W\subset \mathbb{R}^{n}$ and $\delta \in \left(
0,1\right) $ if there exists a function $w:M\rightarrow W$ such that, for
all $i\in I$, 
\begin{align*}
& \left( 1-\delta \right) g_{i}\left( a\right) +\delta E[w_{i}\left( \cdot
\right) |(a,\rho )]-(1-\delta )\eta \\
& \geq \left( 1-\delta \right) g_{i}\left( a_{i}^{\prime },a_{-i}\right)
+\delta E[w_{i}\left( \cdot \right) |(a_{i}^{\prime },\rho _{i}^{\prime
}),(a_{-i},\rho _{-i})]\ \forall (a_{i}^{\prime },\rho _{i}^{\prime })\in 
\hat{\Sigma}_{i}^{(a,\rho ),w}
\end{align*}
\end{definition}
where $\hat{\Sigma}_{i}^{(a,\rho ),w}$ is the set of nontrivial deviations
from $(a,\rho )$ for player $i$ with respect to $w$. If $v=\left( 1-\delta
\right) g\left( a\right) +\delta E[w\left( \cdot \right) |(a,\rho )]$ for
some $\eta $-enforceable pair $(a,\rho )$ and $w:M\rightarrow W$, then we
say that $v$ is $\eta $-\emph{decomposable} and that $((a,\rho ),w)$ $\eta $%
-decomposes $v$ with respect to $W$ and $\delta $. Define the set of $\eta $%
-decomposable payoffs with respect to $W$ and $\delta $ as follows:%
\begin{equation*}
B\left( \delta ,W,\eta \right) :=\{v\in \mathbb{R}^{n}|v\text{ is }\eta\text{-decomposable with respect to }W\text{ and }\delta \}.
\end{equation*}%
We say that $W$ is $\eta$-\emph{self decomposable} with respect to $\delta $
if $W\subset B\left(\delta ,W,\eta \right) $.

It is easy to see that a \textquotedblleft uniformly
strict\textquotedblright\ version of Theorem 1 in Abreu, Pearce, and
Stacchetti \cite{aps} holds here: if $W$ is $\eta$-self decomposable with
respect to $\delta $, then every $v\in B\left( \delta ,W,\eta \right)$ can
be supported by some $\eta$-USPPE. Since the following lemma follows easily
from the result in Abreu, Pearce and Stacchetti, its proof is omitted.

\begin{lemma}
\label{rec}If a nonempty set $W\subset \mathbb{R}^{n}$ is bounded and $\eta$%
-self decomposable with respect to $\delta \in \left( 0,1\right) $, then $%
B\left( \delta ,W,\eta \right) \subset E^{\eta }\left( \delta \right) .$
\end{lemma}

For the rest of this subsection, we prove that local $\eta$-self
decomposability of $W$ implies $\eta$-self decomposability of $W$. In the
framework of repeated games with imperfect public monitoring, Fudenberg,
Levine, and Maskin (\cite{flm}) introduced a notion of local self
decomposability that is sufficient for self decomposability. Here we prove
the corresponding lemma in our setting. We begin with a lemma that
establishes a certain monotonicity property of $B.$ It implies that, if $W$
is $\eta $-self decomposable with respect to $\delta \in \left( 0,1\right) ,$
then $W$ is $\eta $-self decomposable for every $\delta ^{\prime }\in \left(
\delta ,1\right) $.

\begin{lemma}
\label{monotone}If $W\subseteq \mathbb{R}^{n}$ is convex and $C\subseteq
B\left( \delta ,W,\eta \right) \cap W$, then $C\subseteq B\left( \delta
^{\prime },W,\eta \right) $ for every $\delta ^{\prime }\in \left( \delta
,1\right) .$
\end{lemma}

\begin{proof}
Suppose that $v\in C.$ Since $v\in B\left( \delta ,W,\eta \right) ,$ $v$ is $%
\eta$-decomposable with respect to $W$ and $\delta ,$ hence there exists a
pair $(\left( a,\rho ),w^{\delta }\right) $ that $\eta$-decomposes $v.$ For
any $\delta ^{\prime }>\delta ,$ define $w^{\delta ^{\prime }}:M\rightarrow
W $ as the following convex combination of $v$ and $w^{\delta }$: 
\begin{equation*}
w^{\delta ^{\prime }}(m)=\frac{\delta ^{\prime }-\delta }{\delta ^{\prime
}\left( 1-\delta \right) }v+\frac{\delta \left( 1-\delta ^{\prime }\right) }{%
\delta ^{\prime }\left( 1-\delta \right) }w^{\delta }(m).
\end{equation*}%
Clearly, $w^{\delta ^{\prime }}(m)\in W$ for each $m\in M$ since $W$ is
convex. Furthermore, we can show that, for every $\delta ^{\prime }\in
\left( \delta ,1\right) ,$ the pair $\left( (a,\rho ),w^{\delta ^{\prime
}}\right) $ $\eta$-decomposes $v$ with respect to $W$ and $\delta ^{\prime
}. $ To see this, first note that, for all $\delta ^{\prime }\in \left(
\delta ,1\right) ,$ and $i\in N,$ 
\begin{eqnarray*}
&&\left( 1-\delta ^{\prime }\right) g_{i}\left( a\right) +\delta ^{\prime }E 
\left[ w_{i}^{\delta ^{\prime }}(\cdot )|(a,\rho )\right] \\
&=&\left( 1-\delta ^{\prime }\right) g_{i}\left( a\right) +\frac{\delta
\left( 1-\delta ^{\prime }\right) }{1-\delta }E\left[ w_{i}^{\delta }(\cdot
)|(a,\rho )\right] +\frac{\delta ^{\prime }-\delta }{1-\delta }v_{i} \\
&=&\frac{1-\delta ^{\prime }}{1-\delta }\left\{ (1-\delta )g_{i}\left(
a\right) +\delta E\left[ w_{i}^{\delta }(\cdot )|(a,\rho )\right] \right\} +%
\frac{\delta ^{\prime }-\delta }{1-\delta }v_{i} \\
&=&v_i.
\end{eqnarray*}%
Next, note that for all $(a_{i}^{\prime },\rho _{i}^{\prime }) \in
\Sigma_i^{(a, \rho), w^{\delta^\prime}} = \Sigma_i^{(a, \rho), w^{\delta}}$
and $i \in N,$%
\begin{eqnarray*}
&&\left( 1-\delta ^{\prime }\right) g_{i}\left( a\right) +\delta ^{\prime
}E[w_{i}^{\delta ^{\prime }}(\cdot )|(a,\rho )]-(1-\delta ^{\prime })\eta \\
&=&\frac{\left( 1-\delta ^{\prime }\right) }{(1-\delta )}\left[ (1-\delta
)g_{i}\left( a\right) +\delta E\left[ w_{i}^{\delta }(\cdot )|(a,\rho )%
\right] \right] +\frac{\delta ^{\prime }-\delta }{\left( 1-\delta \right) }%
v_{i}-(1-\delta ^{\prime })\eta \\
&\geq &\frac{\left( 1-\delta ^{\prime }\right) }{(1-\delta )}\left[ \left(
1-\delta \right) g_{i}\left( a_{i}^{\prime },a_{-i}\right) +\delta
E[w_{i}^{\delta }\left( \cdot \right) |(a_{i}^{\prime },\rho _{i}^{\prime
}),(a_{-i},\rho _{-i})]+(1-\delta )\eta \right] +\frac{\delta ^{\prime
}-\delta }{\left( 1-\delta \right) }v_{i}-(1-\delta ^{\prime })\eta \\
&=&\left( 1-\delta ^{\prime }\right) g_{i}\left( a_{i}^{\prime
},a_{-i}\right) +\delta ^{\prime }\left[ \frac{\left( 1-\delta ^{\prime
}\right) \delta }{\delta ^{\prime }(1-\delta )}E[w_{i}^{\delta }\left( \cdot
\right) |(a_{i}^{\prime },\rho _{i}^{\prime }),(a_{-i},\rho _{-i})]+\frac{%
\delta ^{\prime }-\delta }{\delta ^{\prime }\left( 1-\delta \right) }v_{i}%
\right] \\
&=&\left( 1-\delta ^{\prime }\right) g_{i}\left( a_{i}^{\prime
},a_{-i}\right) +\delta ^{\prime }E\left[ w_{i}^{\delta ^{\prime }}(\cdot
)|(a_{i}^{\prime },\rho _{i}^{\prime }),(a_{-i},\rho _{-i})\right] .
\end{eqnarray*}

Consequently, $\left( (a,\rho ),w^{\delta ^{\prime }}\right) $ $\eta$%
-decomposes $v$ with respect to $W$ and $\delta ^{\prime }.$ Hence $%
C\subseteq B\left( \delta ^{\prime },W,\eta \right)$ for every $\delta
^{\prime }\in \left( \delta ,1\right) $ and this completes the proof.
\end{proof}

\bigskip

Next we introduce local $\eta$-self decomposability and show that local $%
\eta$-self decomposability implies $\eta$-self decomposability for
sufficiently large discount factors.

\begin{definition}
A nonempty set $W\subseteq \mathbb{R}^{n}$ is \emph{locally} $\eta$-self 
decomposable if, for any $v \in W,$ there exists $\delta \in \left(
0,1\right) $ and an open set $U$ containing $v$ such that $U\cap W\subset
B\left( \delta ,W,\eta \right) .$
\end{definition}

\begin{lemma}
\label{local} If $W\subset \mathbb{R}^{n}$ is compact, convex, and locally $\eta$-self decomposable, 
then there exists a $\underline{\delta }\in
\left( 0,1\right) $ such that $W$ is $\eta$-self decomposable with respect
to $\delta $ for any $\delta \in \left( \underline{\delta },1\right). $
\end{lemma}

\begin{proof}
Choose $v\in W.$ Since $W$ is $\eta $-locally self decomposable, there
exists $\delta _{v}\in \left( 0,1\right) $ and an open ball $U_{v}$ around $%
v $ such that 
\begin{equation*}
U_{v}\cap W\subseteq B\left( \delta _{v},W,\eta \right) .
\end{equation*}%
Since $W$ is compact, there exists a finite sub-collection $\left\{
U_{v_{k}}\right\} _{k=1}^{K}$ that covers $W.$ Define $\underline{\delta }%
=\max_{k=1,...,K}\left\{ \delta _{v_{k}}\right\} $. Then 
\begin{equation*}
U_{v_{k}}\cap W\subseteq B\left( \delta _{v_{k}},W,\eta \right) \subseteq
B\left( \delta ,W,\eta \right)
\end{equation*}%
for any $\delta \in (\underline{\delta },1)$ by Lemma \ref{monotone} and the
convexity of $W$. Consequently, 
\begin{equation*}
W\mathbf{=\cup }_{k=1}^{K}\left( U_{v_{k}}\cap W\right) \subseteq B\left(
\delta ,W,\eta \right) .
\end{equation*}%
for every $\delta \in \left( \underline{\delta },1\right) .$
\end{proof}

\subsubsection{Local Decomposability of a Smooth Set in The Bounding Set}

We call a nonempty compact and convex set in $\mathbb{R}^{n}$ \textit{smooth}
if there exists a unique supporting hyperplane at every boundary point of
the set.

\bigskip

The following lemma shows that, if $Q^{\eta }$ has an interior point in $%
\mathbb{R}^{n}$, then there exists a smooth, compact and convex set in $%
intQ^{\eta }$ that is arbitrarily close to $Q^{\eta }$.

\begin{lemma}
\label{HD} Suppose that $Q^{\eta }\subseteq \mathbb{R}^{n}$ has an interior
point. Then, for every $\varepsilon >0,$ there exists a smooth compact and
convex set $W^{\prime }\subset intQ^{\eta }$ such that the Hausdorff
distance between $W^{\prime }$ and $intQ^{\eta }$ is at most $\varepsilon $.
\end{lemma}

\begin{proof}
Choose any $\varepsilon >0.$ Since bounded sets in Euclidean space are
totally bounded, there exists a finite set $Z\subseteq intQ^{\eta }$ such
that, for each $v \in intQ^{\eta },$ there exists $z \in Z$ such that $%
\left\Vert z-x \right\Vert<\varepsilon .$ Let $W=coZ$. Then $W$ is nonempty,
compact and convex. Since $Q^{\eta }$ is convex, it follows that $intQ^{\eta
}$ is convex, hence $W \subseteq intQ^{\eta }.$ For each $v\in intQ^{\eta }$%
, there exists $z\in W$ such that $\left\Vert z-x \right\Vert<\varepsilon $,
which implies $\sup_{v\in intQ^{\eta }}[\min_{z\in W} \left\Vert z-v
\right\Vert]\leq \varepsilon$. On the other hand, $\max_{z \in W}[\min_{v
\in int Q^\eta}\left\Vert z-v \right\Vert]$ is clearly $0$ since $W
\subseteq intQ^{\eta }$. Hence the Hausdorff distance between $W$ and $int
Q^\eta$ is at most $\varepsilon$.

Next we construct $W^\prime$ from $W$. Since $W \subseteq intQ^{\eta }$ is a
polyhedron and has only a finite number of vertices, we can find a small
enough $\epsilon^{\prime }>0$ such that, at every $v\in W$, the closed ball $%
B_{v}^{\epsilon ^{\prime }}$ of radius $\epsilon^{\prime }$ around $v$ is in 
$intQ^{\eta }$. Define $W^{\prime }=\bigcup_{v\in W}B_{v}^{\epsilon ^{\prime
}}$. Since $W \subseteq W^{\prime } \subseteq intQ^{\eta }$, the Hausdorff
distance between $W^\prime$ and $int Q^\eta$ is at most $\varepsilon$.

Next we show that $W^{\prime }$ is a smooth compact convex set. To show that 
$W^{\prime }$ is compact, it suffices to show that $W^{\prime }$ is closed.
Suppose that $w_{k} \in W^{\prime }$ for each $k$ and $\{w_{k}\}$ is
convergent with limit $w^{\ast }.$ For each $k$, there exists $v_{k}\in W$
such that $\left\Vert v_{k}-w_{k} \right \Vert\leq \epsilon ^{\prime }.$
Since $W$ is compact, we may assume wlog that $\{v_{k}\}$ is convergent with
limit $v^{\ast }\in W. $ Consequently, $\left\Vert v^{\ast }-w^{\ast }
\right\Vert \leq \epsilon ^{\prime }$ implying that $w^{\ast }\in W^{\prime
}.$

To show that $W^{\prime }$ is convex, choose any $x,y\in W^{\prime }$. Then
there exists $x^{\prime },y^{\prime }\in W$ such that $\left\Vert x^{\prime
}-x\right\Vert \leq \epsilon ^{\prime }$ and $\left\Vert y^{\prime
}-y\right\Vert \leq \epsilon ^{\prime }$ by definition of $W^{\prime }$. For
any $\alpha \in \lbrack 0,1]$, $\alpha x^{\prime }+(1-\alpha )y^{\prime }$
is in $W$ and the distance between $\alpha x^{\prime }+(1-\alpha )y^{\prime
} $ and $\alpha x+(1-\alpha )y$ is less than $\alpha \left\Vert x^{\prime
}-x\right\Vert +(1-\alpha )\left\Vert y^{\prime }-y\right\Vert \leq \epsilon
^{\prime }$. So, $\alpha x+(1-\alpha )y\in B_{\alpha x^{\prime }+(1-\alpha
)y^{\prime }}^{\epsilon ^{\prime }}\subseteq W^{\prime }$. Therefore $%
W^{\prime }$ is convex.

Finally, to see that $W^{\prime }$ has a unique supporting hyperplane at
every boundary point, first note that every boundary point of $W^{\prime }$
must be a boundary point of $B_{v}^{\epsilon ^{\prime }}$ for some $v\in W$.
Since a supporting hyperplane of a boundary point of $W^{\prime }$ must be a
supporting hyperplane of $B_{v}^{\epsilon ^{\prime }}$ at the same point and 
$B_{v}^{\epsilon ^{\prime }}$ cannot have multiple supporting hyperplanes at
any boundary point, the supporting hyperplane must be unique at every
boundary point of $W^{\prime }$.
\end{proof}

\bigskip

We now show that any such set $W^\prime$ that approximates $Q^{\eta }$ from
the inside is $\eta$-locally decomposable, which leads to our main result by
Lemma \ref{local}.

We need two technical lemmas for local decomposability.

\begin{lemma}
A smooth, compact and convex set $C\subseteq \mathbb{R}^{n}$ has non-empty
interior in $\mathbb{R}^{n}$.
\end{lemma}

\begin{proof}
Suppose otherwise. Then the affine hull of $C$ has dimension less than $n$.
Let $S$ denote the affine hull and, translating if necessary, we may assume
that $0\in C$ and $S$ is a vector subspace of $\mathbb{R}^{n}.$ Since $C$ is
smooth, $C$ is not a singleton set. So we can find $p\neq 0\in C$. Let $%
x^{\ast }\in \arg \max_{x\in C}p\cdot x.$ Then $p\cdot x\leq p\cdot x^{\ast
} $ for all $x\in C$. Hence the set $\{x\in \mathbb{R}^{n}|p\cdot x=p\cdot
x^{\ast }\}$ is a supporting hyperplane for $C$ at $x^{\ast }$. Now choose $%
q\in S^{\bot }.$ Then $q\cdot x=0$ for all $x\in C$, so the set $\{x\in 
\mathbb{R}^{n}|q\cdot x=0\}$ is a supporting hyperplane for $C$ at $x^{\ast
} $. Clearly $\{x\in \mathbb{R}^{n}|p\cdot x=p\cdot x^{\ast }\}$ and $\{x\in 
\mathbb{R}^{n}|q\cdot x=0\}$ are distinct hyperplanes because the former
does not include $0$ (since $p\cdot x^{\ast }\geq p\cdot p>0)$, while the
latter does.
\end{proof}

\bigskip

\begin{lemma}
\label{smooth} Let $W\subset \mathbb{R}^{n}$ be a smooth, compact and convex
set and let $v$ be a boundary point of $W$. Let $\lambda ^{v}\neq 0\in 
\mathbb{R}^{n}$ be a normal to the unique supporting hyperplane of $W$ at $v,
$ i.e., $\lambda ^{v}\cdot v\geq \lambda ^{v}\cdot x$\ for all $x\in W$.
Then, for any point $y\in \mathbb{R}^{n}$\ such that $\lambda ^{v}\cdot
v>\lambda ^{v}\cdot y,$ there exists $\alpha ^{\ast }\in \left( 0,1\right) $
such that $\left( 1-\alpha \right) v+\alpha y \in int W$ for any $\alpha \in
\left( 0,\alpha ^{\ast }\right) .$

Conversely, if $\lambda^{v}\in \mathbb{R}^{n}$ satisfies $\lambda^{v}\cdot
v\geq \lambda ^{v}\cdot x$ for all $x\in W $ and, for any $y\in \mathbb{R}%
^{n}$ such that $\lambda ^{v}\cdot v>\lambda ^{v}\cdot y$, there exists $%
\alpha ^{\ast }\in \left( 0,1\right) $ such that $\left( 1-\alpha \right)
v+\alpha y \in int W$ for any $\alpha \in \left( 0,\alpha ^{\ast}\right)$,
then there is the unique supporting hyperplane of $W$ at $v$ and $\lambda^v$
is its normal vector.
\end{lemma}

\begin{proof}
Translating if necessary, we may assume that $v=0$. We argue by
contradiction. Suppose that there exists $y\in \mathbb{R}^{n}$ such that $%
\lambda ^{v}\cdot y<0$ but for each $\alpha ^{\ast }\in \left( 0,1\right) $
there exists $\alpha \in \left( 0,\alpha ^{\ast }\right) $ such that $\alpha
y$\ $\notin intW.$ Then there exists a sequence $\{\alpha _{k}\}$ such that $%
0<\alpha _{k}<1,$ $\alpha _{k}\rightarrow 0$ and $\alpha _{k}y\notin intW$
for each $k$. Since $intW$ is non-empty by the previous lemma and convex,
there exists for each $k$ a $q_{k}\neq 0$ such that 
\begin{equation*}
\frac{q_{k}}{||q_{k}||}\cdot x\leq \frac{q_{k}}{||q_{k}||}\cdot \left(
\alpha _{k}y\right)
\end{equation*}%
for all $x\in intW$ by the separating hyperplane theorem. Let $k\rightarrow
\infty $ and $\frac{q_{k}}{||q_{k}||}\rightarrow q$ for some $q\neq 0$,
extracting a subsequence if necessary. Then $q\cdot x\leq 0$ for all $x\in
intW.$ Then $q\cdot x\leq 0$ for all $x\in W$, hence $q$ is a normal vector
for a supporting hyperplane of $W$ at $v=0$. To derive a contradiction, we
show that $q\neq \beta \lambda ^{v}$ for all $\beta >0.$ To see this, take
any $z\in intW$ and note that $\alpha _{k}^{2}z\in intW$ (since $0\in W).$
Therefore 
\begin{equation*}
\frac{q_{k}}{||q_{k}||}\cdot (\alpha _{k}z)\leq \frac{q_{k}}{||q_{k}||}\cdot
y
\end{equation*}%
implying that $0\leq q\cdot y.$ If $\beta >0,$ then $(\beta \lambda
^{v})\cdot y<0.$ Consequently, $q\neq \beta \lambda ^{v}$, which contradicts
the smoothness of $W$.

For the converse, $\left\{x \in \mathbb{R}^n| \lambda^v \cdot x = 0\right\}$
is clearly supporting hyperplane of $W$ at $v=0$. Suppose that there is a
different supporting hyperplane $\left\{x \in \mathbb{R}^n| \lambda^\prime
\cdot x = 0\right\}$ of $W$ at $v =0$, which satisfies $\lambda^\prime \cdot
x \leq 0$ for any $x \in W$. Then there must exist $y^\prime \in \mathbb{R}^n
$ such that $\lambda^\prime \cdot y^\prime > 0$ and $\lambda ^v \cdot
y^\prime < 0$. Then $\alpha y^\prime$ is in $W$ for small enough $\alpha \in
(0,1)$ by assumption, but $\lambda^\prime \cdot (\alpha y^\prime) > 0$,
which is a contradiction. Hence there is the unique supporting hyperplane of 
$W$ at $v=0$ and $\lambda^v$ is its normal vector.
\end{proof}

\begin{theorem}
\label{limit} Suppose that $Q^{\eta }$ has an interior point in $\mathbb{R}%
^{n}$. For any $\epsilon >0$, there exists a smooth, compact and
convex set $W\subseteq intQ^{\eta }$ and $\underline{\delta }\in (0,1)$ such
that $W \subset E^{\eta }(\delta )$ for any $\delta \in (\underline{\delta }%
,1)$ and the Hausdorff distance between $W$ and $Q^{\eta }$ is at most $%
\epsilon $.
\end{theorem}

\begin{proof}
By Lemma \ref{HD}, there exists a smooth, compact and convex set $%
W $ in $intQ^{\eta }$ such that the Hausdorff distance between $W$ and int$%
Q^{\eta }$ is at most $\epsilon $. Since $cl(intQ^{\eta })=Q^{\eta },$ it
follows that the Hausdorff distance between $W$ and $Q^{\eta }$ is at most $%
\epsilon $. By Lemma \ref{local}, it suffices to show that $W$ is locally $\eta$-self decomposable. 
Take any boundary point $w^{\ast }$ of $W$. Since 
$W$ is smooth, there is unique $\lambda ^{\ast }\in \Lambda $ such that $%
\lambda ^{\ast }\cdot w^{\ast }=\max_{w\in W}\lambda ^{\ast }\cdot w$. Since 
$W\subseteq intQ^{\eta }$, there exists a feasible solution $(v,(a,\rho ),x)$
for the programming problem $\left( P^{\lambda ^{\ast },\eta }\right) $ such
that $\lambda ^{\ast }\cdot v>\lambda ^{\ast }\cdot w^{\ast }$. Define $%
x^{\prime }(m)=x(m)-(v-w^{\ast })$ for each $m\in M$. Then $(w^{\ast
},(a,\rho ),x^{\prime })$ is feasible in the programming problem because the 
$\eta $-strict incentive constraints for nontrivial deviations are not
affected, $\lambda ^{\ast }\cdot x^{\prime }(m)=\lambda ^{\ast }\cdot
x(m)-\lambda ^{\ast }\cdot (v-w^{\ast })<0$ for each $m\in M$, and 
\begin{align*}
& g\left( a\right) +E[x^{\prime }(\cdot )|(a,\rho )] \\
& =g\left( a\right) +\sum_{m\in M}x(m)\tilde{p}\left( m|(a,\rho )\right)
-v+w^{\ast } \\
& =w^{\ast }
\end{align*}

For each $\delta $ and $m$, define $w^{\delta }(m)$ by $w^{\delta
}(m):=w^{\ast }+\frac{1-\delta }{\delta }x^{\prime }(m)$. If $(a,\rho )$ is
played and $w^{\delta }$ is used as the continuation payoff profile, then,
for all $(a_{i}^{\prime },\rho _{i}^{\prime })\in \hat{\Sigma}_{i}^{(a,\rho
),w^{\delta }}=\hat{\Sigma}_{i}^{(a,\rho ),x^{\prime }}$ and $i\in N$,
\begin{align*}
& \left( 1-\delta \right) g_{i}\left( a\right) +\delta E[w_{i}^{\delta
}\left( \cdot \right) |(a,\rho )]-(1-\delta )\eta \\
& =\left( 1-\delta \right) \left( g_{i}\left( a\right) +E\left[
x_{i}^{\prime }(\cdot )|(a,\rho )\right] \right) +\delta w_{i}^{\ast
}-(1-\delta )\eta \\
& \geq \left( 1-\delta \right) \left( g_{i}\left( a_{i}^{\prime
},a_{-i}\right) +E\left[ x_{i}^{\prime }(\cdot )|(a_{i}^{\prime },\rho
_{i}^{\prime }),(a_{-i},\rho _{-i})\right] +\eta \right) +\delta w_{i}^{\ast
}-(1-\delta )\eta \\
& =\left( 1-\delta \right) \left( g_{i}\left( a_{i}^{\prime },a_{-i}\right)
+E\left[ \frac{\delta }{1-\delta }\left( w_{i}^{\delta }(\cdot )-w_{i}^{\ast
}\right) |(a_{i}^{\prime },\rho _{i}^{\prime }),(a_{-i},\rho _{-i})\right]
\right) +\delta w_{i}^{\ast } \\
& =\left( 1-\delta \right) g_{i}\left( a_{i}^{\prime },a_{-i}\right) +\delta
E\left[ w_{i}^{\delta }(\cdot )|(a_{i}^{\prime },\rho _{i}^{\prime
}),(a_{-i},\rho _{-i})\right] .
\end{align*}
Also note that $(1-\delta )g(a)+\delta E[w^{\delta }(\cdot )|(a,\rho
)]=w^{\ast }$.

Next we show that $w^{\delta }(m)$ is in $intW$ for every $m$ if $\delta $
is large enough. Since $W$ is smooth, $\lambda^{\ast }$ is a normal vector
of the unique supporting hyperplane of $W$ at $w^{\ast}$. Choose any $%
\delta^{\prime }\in (0,1)$. Since $\lambda ^{\ast }\cdot x^{\prime }(m)<0$,
it follows that $\lambda ^{\ast }\cdot w^{\delta ^{\prime }}(m)<\lambda
^{\ast }\cdot w^{\ast }$ for each $m\in M$. Since $\frac{\delta^\prime}{%
1-\delta^\prime}(w^{\delta^{\prime}}(m) - w^{\ast }) = \frac{\delta }{%
1-\delta }(w^{\delta}(m) - w^{\ast })$ for any $\delta, \delta^\prime \in
(0,1)$ by definition, we have 
\begin{equation*}
w^{\delta }(m)=\left( 1-\frac{(1-\delta )\delta ^{\prime }}{\delta (1-\delta
^{\prime })}\right) w^{\ast }+\frac{(1-\delta )\delta ^{\prime }}{\delta
(1-\delta ^{\prime })}w^{\delta ^{\prime }}(m)
\end{equation*}%
for $\delta \in (\delta^{\prime },1)$. Then, for each $m$, there exists $%
\underline{\delta }_{m}\in (0,1)$ such that $w^{\delta }(m)\in int(W)$ for
any $\delta \in (\underline{\delta }_{m},1)$ by Lemma \ref{smooth}. Let $%
\underline{\delta }=\max_{m}\underline{\delta }_{m}$. Then $((a,\rho
),w^{\delta })$ $\eta$-decomposes $w^{\ast }$ with respect to $int(W)$ and $%
\delta $ for any $\delta >\underline{\delta }$.

\bigskip

For each $\xi \in \mathbb{R}^{n}$, let $f^{\delta }(\xi )=(1-\delta
)g(a)+\delta E[w^{\delta }(\cdot )+\xi |(a,\rho )].$ Then $f^{\delta }$ is
continuous, injective and $f(0)=w^{\ast }.$ Since $w^{\delta }(m)\in intW$
for each $m\in M$ and $M$ is finite, there exists an open neighborhood $%
V^{\delta }$ of $0$ such that $w^{\delta }(m)+\xi \in intW$ for each $m\in M$
and each $\xi \in V^{\delta }$ and $f^{\delta }(V^{\delta })=U^{\delta }$ is
an open neighborhood of $w^{\ast }.$ Since $f^{\delta }$ maps V$^{\delta }$
homeomorphically onto $U^{\delta }$, it follows that every point $u\in
U^{\delta }$ can be $\eta$-decomposed by $((a,\rho ),w^{\delta }+(f^{\delta
})^{-1}(u))$ with respect to $int(W)$ for each $\delta >\underline{\delta }$%
. Therefore, $U^{\delta }\cap W\subseteq U^{\delta }\subseteq B\left( \delta
,intW,\eta \right) \subseteq B\left( \delta ,W,\eta \right) $ for any $%
\delta >\underline{\delta }$. A similar argument applies for any $w^\ast \in
intW.$ Hence $W$ is locally $\eta $-self decomposable.
\end{proof}

\section{Uniformly Strict Folk Theorem}

In this section, we prove a folk theorem with $\eta$-uniformly strict PPE by
showing that $Q^\eta$ coincides with $V^*(G)$ under certain conditions. In
the following, we use $\phi_{i}$ to denote a mixed strategy over $%
A_{i}\times R_{i}$. Let $\alpha _{i}(\phi _{i})$ denote the marginal
distribution of $\phi _{i}$ on $A_{i}$. For each $(a_{-i},\rho _{-i}),$ let 
\begin{equation*}
\tilde{p}(\cdot |\phi _{i},(a_{-i},\rho _{-i}))=\sum_{(a_{i}^{\prime },\rho
_{i}^{\prime })\in A_{i}\times R_{i}}\tilde{p}(\cdot | (a_{i}^{\prime },\rho
_{i}^{\prime }),(a_{-i},\rho _{-i}))\phi _{i}\left((a_{i}^{\prime },\rho
_{i}^{\prime })\right)
\end{equation*}
and let $\tilde{p}_{-i}(\cdot |\phi _{j},(a_{-j},\rho _{-j}))$ be the
marginal distribution of $\tilde{p}(\cdot |\phi _{j},(a_{-j},\rho _{-j}))$
over $M_{-i}$.

We need the following four conditions on the private monitoring structure
and the payoff functions for our folk theorem. Recall that $A(G)\subseteq A$
is the set of action profiles that generate an extreme point in $V(G).$

\begin{definition}[$\protect\eta$-detectability]
For each $a \in A(G)$, there exists $\rho \in R$ that satisfies the
following condition: for each $i \in N$, if $\tilde{p}(\cdot |\phi
_{i},(a_{-i},\rho _{-i}))=\tilde{p}(\cdot |(a,\rho ))$ for some $\phi
_{i}\in \Delta\left((A_i \times R_i)\backslash \{(a_i, \rho_i)\}\right)$,
then $g_{i}(a)-g_{i}(\alpha _{i}(\phi _{i}),a_{-i})\geq \eta $ holds.
\end{definition}

This condition means that if player $i$'s unilateral deviation to a mixed
strategy (with $0$ probability on $(a_i, \rho_i)$) cannot be detected, then
she must lose at least $\eta$ in terms of the stage-game expected payoff

When we approximate the minmax point, we need a slightly stronger
detectability condition.

\begin{definition}[$\protect\eta^*$-detectability with respect to $i$ at $a
\in A$]
There exists $\rho \in R$ that satisfies the following condition: for each $%
j \neq i$, if $\tilde{p}_{-i}(\cdot |\phi _{j},(a_{-j},\rho _{-j}))= \tilde{p%
}_{-i}(\cdot |(a,\rho ))$ for some $\phi _{j} \in \Delta\left((A_j \times
R_j)\backslash\{(a_j, \rho_j)\}\right)$, then $g_{j}(a)-g_{j}(\alpha
_{j}(\phi _{j}),a_{-j}) \geq \eta$ holds.
\end{definition}

This condition means that the above $\eta$-detectability condition holds for 
$j \neq i$ without using player $i$'s message.

The next condition means that, if player $i$'s deviation is not linearly
independent from some other player's deviation, then she must lose at least $%
\eta$ in terms of the stage-game expected payoff.

\begin{definition}[$\protect\eta$-identifiability]
For each $a \in A(G)$, there exists $\rho \in R$ that satisfies the
following condition: for each pair $i \neq j$, if $\tilde{p}(\cdot |\phi
_{i},(a_{-i},\rho_{-i}))-\tilde{p}(\cdot |(a,\rho ))$ and $\tilde{p}(\cdot
|\phi_{j},(a_{-j},\rho _{-j}))-\tilde{p}(\cdot |(a,\rho ))$ are not linearly
independent for some $\phi _{i} \in \Delta\left((A_i \times
R_i)\backslash\{(a_i, \rho_i)\}\right)$ and $\phi _{j} \in
\Delta\left((A_j\times R_j)\backslash\{(a_j, \rho_j)\}\right)$, then $%
\min\{g_{i}(a)-g_{i}(\alpha _{i}(\phi _{i}),a_{-i}), g_j(a)-g_{j}(\alpha
_{j}(\phi _{j}),a_{-j}) \} \geq \eta $ holds.
\end{definition}

The last conditions require that, for every player $i$, there exists the
best action profile and the minmax action profile, where player $i$ would
lose at least $\eta$ by deviating to any other pure action. Remember Lemma %
\ref{bound} and our discussion following the lemma; we know that they are
necessary for the folk theorem.

\begin{definition}[$\protect\eta$-best response property]
$G$ satisfies $\eta$-best response property for $\left\{\overline{a}^i, 
\underline{a}^i, i \in N \right\} \subset A$ if the following conditions are
satisfied for any $i \in N$:

\begin{enumerate}
\item $g_i(\overline{a}^i) =\max_{a}g_i(a)$ and $g_i(\overline{a}^i) -
g_i(a_i^\prime, \overline{a}_{-i}^{i}) \geq \eta$ for any $a_i^\prime \neq 
\overline{a}^i_i$.

\item $g_i(\underline{a}^i) = \min_{a_{-i}} \max_{a_i} g_i(a_i, a_{-i})$ and 
$g_i(\underline{a}^i) - g_i(a_i^\prime, \underline{a}_{-i}^{i}) \geq \eta$
for any $a_i^\prime \neq \underline{a}^i_i$.
\end{enumerate}
\end{definition}

It may be useful to compare these conditions to the similar conditions
(A1)-(A3) for Theorem 1 in \cite{km}. (A1) requires $0^*$-detectability
condition at the minmax action profile. We instead assume $\eta^*$%
-detectability condition at the minmax action profile and the best action
profile for each player. Kandori and Matsushima \cite{km}
assumes (A2) and (A3) for every action profile in $A(G)$, which is a
restriction on the distribution of the private signals of any subset of $n-2$
players, whereas we assume $\eta$-detectability and $\eta$-identifiability
for every $a \in A(G)$, which is a restriction on the joint distribution of all messages. 
$\eta$-detectability and $\eta$-identifiability are weaker than
(A2) and (A3) when the message space is rich enough in the following sense.
For $\eta$-identifiability, if (A2) and (A3) are satisfied at $a$ and $M_i
=S_i$ for every $i \in N$, then $\eta$-identifiability is automatically satisfied
with truthful message strategies, since the type of linear dependency that
appears in the definition of $\eta$-identifiability would never occur given
(A2) and (A3). For the same reason, (A2) implies $\eta$-detectability for every $a \in A(G)$.

As an example of monitoring structure that satisfies our conditions,
consider $p$ that satisfies the individual full rank condition for each
player with respect to the other players' signals and the pairwise full rank
condition for every pair of players with respect to the private signals of
the other $n-2$ players. Then (A2) and (A3) are satisfied. In addition, $\eta^*$-detectability is trivially satisfied. Hence all our conditions on the monitoring structure ($\eta^*$-detectability and $\eta$-detectability \& $\eta$-identifiability) are satisfied.

Using these conditions, we can state our folk theorem with $\eta$-USPPE using as follows.

\begin{theorem}
\label{folk} Fix any private monitoring game $\left(G, p\right)$. Suppose
that $int V^*(G) \neq \emptyset$ and $G$ satisfies $\eta$-best response
property for $\left\{\overline{a}^i, \underline{a}^i, i \in N \right\}
\subset A$. If $\left(G, p\right)$ satisfies both $\eta$-detectability and $%
\eta$-identifiability with the same $\rho^a \in R$ for each $a \in A(G)$ and
satisfies $\eta^*$-detectability at $\overline{a}_i$ and $\underline{a}_i$
with respect to $i$ for every $i \in N$, then, for any $\epsilon > 0$, there
exists a smooth, compact and convex set $W \subseteq int V^*(G)$ and $%
\underline{\delta }\in (0,1)$ such that $W\subset E^{\eta}(\delta )$ for any 
$\delta \in (\underline{\delta}, 1)$ and the Hausdorff distance between $W$
and $V^*(G)$ is at most $\epsilon$.
\end{theorem}

We prove this theorem through a series of lemma. We first observe that $\eta$%
-detectability is equivalent to the existence of a transfer $x$ that
guarantees $\eta$-strict incentive compatibility.

\begin{lemma}
\label{Farkas} $(G, p)$ satisfies $\eta$-detectability if and only if for
any $a \in A(G)$, there exists $\rho$ such that there exists $x_i: M
\rightarrow \mathbb{R}$ for each $i \in N$ that satisfies 
\begin{align*}
& g_{i}\left( a \right) +\frac{\delta }{1-\delta }\sum_{m\in M}x_i\left(
m\right) \tilde{p}\left( m|(a,\rho)\right) -\eta \geq \\
& g_{i}\left( a_{i}^{\prime },a_{-i}\right) +\frac{\delta }{1-\delta }%
\sum_{m\in M}x_i(m)\tilde{p}(m|(a_{i}^{\prime },\rho _{i}^{\prime
}),(a_{-i},\rho _{-i}))\ \forall (a_{i}^{\prime }, \rho_i^\prime) \neq (a_i,
\rho_i)
\end{align*}
\end{lemma}

Since the proof for this result is standard, it is omitted.\footnote{%
For example, see the proof of the the corresponding result in \cite{km} (p.
650).} By the same argument, we can show that $\eta^*$-detectability with
respect to $i$ is equivalent to the existence of a transfer that does not
depend on player $i$'s message and guarantees $\eta$-strict incentive for
every player other than $i$.

\begin{lemma}
\label{Farkas2} $(G,p)$ satisfies $\eta ^{\ast}$-detectability for $i$ with
respect to $a\in A$ if and only if there exists a $\rho$ and, for each $%
j\neq i$, a function $x_{j}:M_{-i}\rightarrow R$ satisfying 
\begin{align*}
& g_{j}\left( a\right) +\frac{\delta }{1-\delta }\sum_{m_{-i}\in
M_{-i}}x_{j}\left( m_{-i}\right) \tilde{p}_{-i}\left( m_{-i}|(a,\rho
)\right) -\eta \geq \\
& g_{j}\left( a_{j}^{\prime },a_{-j}\right) +\frac{\delta }{1-\delta }%
\sum_{m_{-i}\in M_{-i}}x_{j}(m_{-i})\tilde{p}_{-i}(m_{-i}|(a_{j}^{\prime
},\rho _{j}^{\prime }),(a_{-j},\rho _{-j}))\ \forall (a_{j}^{\prime },\rho
_{j}^{\prime })\neq (a_{j},\rho _{j})
\end{align*}
\end{lemma}

The next lemma shows that $k^\eta(\lambda)$ is equal to $\max_{a \in A}
\lambda \cdot g(a)$ for any regular $\lambda$ (with at least two nonzero
elements) when $\eta$-detectability and $\eta$-identifiability are
satisfied.

\begin{lemma}
\label{sublemma1} Suppose that $(G,p)$ satisfies $\eta$-detectability and $%
\eta$-identifiability with the same $\rho \in R$ for each $a \in A(G)$.
Then, for any $\lambda \in \Lambda \notin \{\pm e^i, i \in N\}$, $%
k^\eta(\lambda) = \max_{a \in A} \sum_i \lambda_i g_i(a)$.
\end{lemma}

\begin{proof}
Pick any $a^\lambda \in A(G)$ that solves $\max_{a \in A} \sum_i \lambda_i
g_i(a)$. By assumption, there is the same $\rho^\lambda \in R$ for which the
conditions for $\eta$-detectability and $\eta$-identifiability are satisfied
at $a^\lambda$. We show that there exists $x:M\rightarrow \mathbb{R}^{n}$
satisfying the following conditions: 
\begin{align*}
\sum_{m\in M}& x_{i}(m)\left( \tilde{p}(m|(a^{\lambda },\rho^\lambda))-%
\tilde{p}(m|(a_{i}^{\prime },\rho _{i}^{\prime }),(a_{-i}^{\lambda },\rho
_{-i}^\lambda))\right) \geq g_{i}(\alpha _{i}^{\prime },a_{-i}^{\lambda
})-g_{i}(a^{\lambda })+\eta \\
\ & \ \forall (a_{i}^{\prime },\rho _{i}^{\prime })\neq (a_{i}^{\lambda
},\rho _{i}^\lambda)\ \forall i\in N \\
\sum_{i}& \lambda _{i}x_{i}(m)=0\ \forall m\in M.
\end{align*}%
This implies $k^{\eta }(\lambda )=\sum_{i}\lambda _{i}g_{i}(a^{\lambda })$
because $((a^\lambda, \rho^\lambda), x)$ is feasible for the problem $%
(P^{\lambda, \eta})$ and achieves the upper bound $\sum_i \lambda_i
g_i(a^{\lambda }) $, hence is clearly a maximum point for $(P^{\lambda,
\eta})$. Note that every on-path deviation is a nontrivial deviation with respect
to transfer $x$ we find. The existence of such $\left(x_i\right)_{i \in N}$
is equivalent to the feasibility of the following linear programming problem
(with value $0$):

\begin{align*}
\min_{x} 0 & \\
\sum_{m \in M}& x_{i}(m)\left(\tilde{p}(m|(a^\lambda,\rho^\lambda)) -\tilde{p%
}((a_i^\prime, \rho_i^\prime),(a^\lambda_{-i},\rho _{-i}^\lambda))\right)
\geq g_i(a_i^\prime, a^\lambda_{-i}) - g_i(a^\lambda) + \eta \\
\ &\ \forall (a_i^\prime, \rho_i^\prime) \neq (a_{i}^\lambda,\rho
_{i}^\lambda) \ \forall i \in N \\
\sum_{i \in N}& \lambda_i x_i(m) = 0 \ \forall m \in M
\end{align*}

The dual problem of this problem is: 
\begin{align*}
\max_{q \geq 0, d} &\sum_{i \in N} \sum_{(a_i^\prime, \rho_i^\prime) \neq
(a_{i}^\lambda,\rho _{i}^\lambda)} \left(g_i(a_i^\prime, a^\lambda_{-i}) -
g_i(a^\lambda) +\eta\right) q_i((a_i^\prime, \rho_i^\prime)) \\
\sum_{(a_i^\prime, \rho_i^\prime) \neq (a_{i}^\lambda,\rho _{i}^\lambda)}
&\left(\tilde p(m|(a^\lambda,\rho^\lambda)) - \tilde{p}(m|(a_i^\prime,
\rho_i^\prime), (a^\lambda_{-i},\rho_{-i}^\lambda)) \right)q_i((a_i^\prime,
\rho_i^\prime)) = \lambda_i d(m)\ \forall m \in M, \ \forall i \in N
\end{align*}
where $q_i((a_i^\prime, \rho_i^\prime)) \geq 0$ is the multiplier for the
strict incentive constraint for $(a_i^\prime, \rho_i^\prime) \neq
(a_i^\lambda, \rho_i^\lambda)$ and $d(m) \in \mathbb{R}$ is the multiplier
for the $\lambda$-``budget balancing'' condition for $m \in M$.

By the strong duality theorem, the value of the primal problem is $0$ if and
only if the value of the dual problem is $0$. Take any $(q,d)$ that is
feasible for the dual problem. For each $i$, we consider two cases. First
suppose $\lambda _{i}=0$. Then the following holds for all $m\in M$: 
\begin{equation*}
\sum_{(a_{i}^{\prime }, \rho _{i}^{\prime })\neq (a_{i}^{\lambda },\rho
_{i}^\lambda)}\left( \tilde{p}(m|(a_{i}^{\prime }, \rho _{i}^{\prime
}),(a_{-i}^{\lambda },\rho _{-i}^\lambda))-\tilde{p}(m|(a^{\lambda
},\rho^\lambda ))\right) q_{i}((a_{i}^{\prime }, \rho _{i}^{\prime }))=0
\end{equation*}%
If $q_{i}((a_{i}^{\prime },\rho _{i}^{\prime }))=0$ for each $(a_{i}^{\prime
}, \rho _{i}^{\prime })\neq (a_{i}^{\lambda },\rho _{i}^\lambda)$, then the $%
i$th term of the objective function is $0$. If $q_{i}((a_{i}^{\prime },\rho
_{i}^{\prime }))\neq 0$ for some $(a_{i}^{\prime }, \rho _{i}^{\prime })\neq
(a_{i}^{\lambda },\rho _{i}^\lambda)$, then this condition is equivalent to $%
\tilde{p}(m|\phi _{i}^{\prime },(a_{-i}^\lambda,\rho _{-i}^\lambda))=\tilde{p%
}(m|(a^\lambda,\rho^\lambda ))\ \forall m\in M$, where $\phi _{i}^{\prime
}\in \Delta ((A_{i}\times R_{i}) \backslash \{(a_{i}^{\lambda },\rho
_{i}^\lambda)\})$ is defined by $\phi _{i}^{\prime }((a_{i}^{\prime }, \rho
_{i}^{\prime }))=\frac{q_{i}((a_{i}^{\prime },\rho _{i}^{\prime }))}{%
\sum_{(a_{i}^{\prime }, \rho _{i}^{\prime })\neq (a_{i}^{\lambda },\rho
_{i}^\lambda)}q_{i}((a_{i}^{\prime }, \rho _{i}^{\prime }))}$. Note that the 
$i$th term of the objective function can be written as 
\begin{equation*}
\left(\sum_{(a_{i}^{\prime }, \rho _{i}^{\prime })\neq (a_{i}^{\lambda
},\rho _{i}^\lambda)}q_{i}((a_{i}^{\prime }, \rho _{i}^{\prime }))
\right)\sum_{(a_{i}^{\prime },\rho _{i}^{\prime })\neq (a_{i}^{\lambda
},\rho _{i}^\lambda)}\left( g_{i}(\alpha _{i}(\phi _{i}^{\prime
}),a_{-i}^{\lambda })-g_{i}(a^{\lambda })+\eta \right)
\end{equation*}%
which is bounded above by $0$ by $\eta$-detectability.

Next suppose that $\lambda _{i} \neq 0.$ Then there exists $j$ such that $%
\lambda _{j}\neq 0$ since $\lambda \notin \{\pm e^{i},i\in N\}.$
Consequently, for all $m\in M$ we have 
\begin{gather*}
\sum_{(a_i^\prime, \rho_i^\prime) \neq (a_{i}^\lambda,\rho
_{i}^\lambda)}\left(\tilde{p}(m|(a^\lambda,\rho^\lambda)) - \tilde{p}%
(m|(a_i^\prime, \rho_i^\prime),(a_{-i}^\lambda,\rho _{-i}^\lambda))\right)
q_{i}((a_i^\prime, \rho_i^\prime))= \lambda_{i}d(m)\text{ } \\
\sum_{(a_j^\prime, \rho_j^\prime) \neq (a_{j}^\lambda,\rho
_{j}^\lambda)}\left(\tilde{p}(m|(a^\lambda,\rho^\lambda)) - \tilde{p}%
(m|(a_j^\prime, \rho_j^\prime), (a_{-j}^\lambda,\rho _{-j}^\lambda))\right)
q_{j}((a_j^\prime, \rho_j^\prime))= \lambda_{j}d(m)\text{ }
\end{gather*}

If $d(m)=0$ for all $m,$ then we can apply the same argument as before to
show that the $i$th and $j$th terms of the objective function are at most $0 
$. If $d(m)\neq 0,$ then $q_{i}$ is not identically 0 nor is $q_{j}$
identically 0. So we can \textquotedblleft cross multiply" the two
equalities, cancel $d(m)$ and conclude that, for all $m\in M$, 
\begin{align*}
& \left[ \frac{\lambda _{j}\left( \sum_{(a_{i}^{\prime },\rho _{i}^{\prime
})\neq (a_{i}^{\lambda },\rho _{i}^\lambda)}q_{i}((a_{i}^{\prime },\rho
_{i}^{\prime }))\right) }{\lambda _{i}\left( \sum_{(a_{j}^{\prime },\rho
_{j}^{\prime })\neq (a_{j}^{\lambda },\rho
_{j}^\lambda)}q_{j}((a_{j}^{\prime },\rho _{j}^{\prime }))\right) }\right]
\left( \tilde{p}(m|\phi _{i}^{\prime },(a_{-i}^{\lambda },\rho
_{-i}^\lambda))-\tilde{p}(m|(a^{\lambda },\rho^\lambda ))\right) \\
& =\left( \tilde{p}(m|\phi _{j}^{\prime },(a_{-j}^{\lambda },\rho
_{-j}^\lambda))-\tilde{p}(m|(a^{\lambda },\rho^\lambda ))\right)
\end{align*}%
where $\phi _{i}^{\prime }$ and $\phi _{j}^{\prime }$ are defined by $\phi
_{i}^{\prime }((a_{i}^{\prime },\rho _{i}^{\prime }))=\frac{%
q_{i}((a_{i}^{\prime },\rho _{i}^{\prime }))}{\sum_{(a_{i}^{\prime },\rho
_{i}^{\prime })\neq (a_{i}^{\lambda },\rho
_{i}^\lambda)}q_{i}((a_{i}^{\prime },\rho _{i}^{\prime }))}$ and $\phi
_{j}^{\prime }((a_{j}^{\prime },\rho _{j}^{\prime }))=\frac{%
q_{j}((a_{j}^{\prime },\rho _{j}^{\prime }))}{\sum_{(a_{j}^{\prime },\rho
_{j}^{\prime })\neq (a_{j}^{\lambda },\rho
_{j}^\lambda)}q_{j}((a_{j}^{\prime },\rho _{j}^{\prime }))}$ respectively.

Since $\tilde{p}(\cdot |\phi _{i}^{\prime },(a_{-i}^{\lambda },\rho
_{-i}^\lambda))-\tilde{p}(\cdot |(a^{\lambda },\rho^\lambda ))$ and $\tilde{p%
}(\cdot |\phi _{j}^{\prime },(a_{-j}^{\lambda },\rho _{-j}^\lambda))-\tilde{p%
}(\cdot |(a^{\lambda },\rho^\lambda))$ are not linearly independent, it
follows from $\eta$-identifiability that both \newline
$\sum_{(a_{i}^{\prime },\rho _{i}^{\prime })\neq (a_{i}^{\lambda },\rho
_{i}^\lambda)}\left( g_{i}(\alpha _{i}(\phi _{i}^{\prime }),a_{-i}^{\lambda
})-g_{i}(a^{\lambda })+\eta \right)$ and $\sum_{(a_{j}^{\prime },\rho
_{j}^{\prime })\neq (a_{j}^{\lambda },\rho _{j}^\lambda)}\left( g_{j}(\alpha
_{j}(\phi _{j}^{\prime }),a_{-j}^{\lambda })-g_{j}(a^{\lambda })+\eta
\right) $ are bounded above by $0$. This implies that the $i$th term (and
the $j$th term) of the objective function are bounded above by $0$.

Hence the $i$th term of the objective function is bounded above by $0$ in
either case for any feasible $(q,d)$, implying that the value of the dual
problem is bounded above by $0$ for any feasible $(q,d)$. Since $0$ can be
achieved by $q(\cdot )=0$ and $d(\cdot )=0$, the value of the dual problem
is exactly $0$ as we wanted to show.
\end{proof}

\bigskip

The next lemma shows that $\eta$-best response property and $\eta^*$
-detectability with respect to $i$ is sufficient to guarantee $k^\eta(e^{i})
= \max_{a} g_i(a)$ and $k^\eta(-e^{i}) = 0$.

\begin{lemma}
\label{sublemma2} Suppose that $G$ satisfies $\eta$-best response property
for $\left\{\overline{a}^i, \underline{a}^i, i \in N \right\} \subset A$.
Then the following holds for each $i \in N$.

\begin{itemize}
\item If $(G, p)$ satisfies $\eta^*$-detectability with respect to $i$ at $%
\overline{a}^i$, then $k^\eta(e^{i}) = \max_{a} g_i(a)$.

\item If $(G, p)$ satisfies $\eta^*$-detectability with respect to $i$ at $%
\underline{a}^i$, then $k^\eta(-e^{i}) = - \min_{a_{-i}}\max_{a_i} g_i(a) =
0 $.
\end{itemize}
\end{lemma}

\begin{proof}
For $\lambda =e^{i}$, we can find $\overline{a}^{i}\in A$ such that $g_{i}(%
\overline{a}^{i})=\max_{a}g_{i}(a)$ and $g_{i}(\overline{a}%
^{i})-g_{i}(a_{i}^{\prime },\overline{a}_{-i}^{i})\geq \eta $ for any $%
a_{i}^{\prime }\neq \overline{a}_{i}^{i}$ by assumption. Let $\overline{\rho 
}^{i}\in R$ be the profile of message strategies for which the conditions
for $\eta ^{\ast }$-detectability with respect to $i$ are satisfied at $%
\overline{a}^{i}$ for any $j\neq i$. By Lemma \ref{Farkas2}, for each $j\neq
i$, there exists $x_{j}:M_{-i}\rightarrow \mathbb{R}$ such that all the $%
\eta $-strict incentive compatibility conditions are satisfied for any $%
(a_{j}^{\prime },\rho _{j}^{\prime })\neq (\overline{a}_{j}^{i},\overline{%
\rho }_{j}^{i})$. For player $i$, set $x_{i}(m)=0$ for all $m\in M$. Then
the $\eta $-strict incentive compatibility conditions for player $i$ are
satisfied for every nontrivial deviation $(a_{i}^{\prime },\rho _{i}^{\prime
})\in \hat{\Sigma}_{i}^{(\overline{a}^{i},\overline{\rho }^{i}),x}$, since $i
$'s message does not affect the transfer $x$ for any player (so deviating in
message after the equilibrium action is a trivial deviation). Since $%
\sum_{i}\lambda _{i}x_{i}(m)=0$ for each $m$  by construction, $(\overline{a}%
^{i},\overline{\rho }^{i},x)$ generates an objective function value of $%
\max_{a}g_{i}(a)$ for the problem $(P^{e^{i},\eta })$. Clearly this is the
largest possible value for $(P^{e^{i},\eta })$, hence $k^{\eta
}(e^{i})=\max_{a}g_{i}(a)$.

For $\lambda =-e^{i}$, we can find $\underline{a}^{i}\in A$ such that $g_{i}(%
\underline{a}^{i})=\min_{a_{-i}}\max_{a_{i}}g_{i}(a)=0$ and $g_{i}(%
\underline{a}^{i})-g_{i}(a_{i}^{\prime },\underline{a}_{-i}^{i})\geq \eta $
for any $a_{i}^{\prime }\neq \underline{a}_{i}^{i}$. Let $\underline{\rho }%
^{i}\in R$ be any profile of message strategies for which the conditions for 
$\eta ^{\ast }$-detectability with respect to $i$ are satisfied at $%
\underline{a}^{i}$ for any $j\neq i$. As in the previous case, we can find $%
x_{j}:M_{-i}\rightarrow \mathbb{R}$ for each $j\neq i$ such that all the $%
\eta $-strict incentive compatibility conditions are satisfied for $j$. Set $%
x_{i}(m)=0$ for all $m$ for player $i$. Since $\sum_{i}\lambda
_{i}x_{i}(m)=0 $, $(\underline{a}^{i},\underline{\rho }^{i},x)$ generates an
objective function value of $-g_{i}(\underline{a}^{i})=-\min_{a_{-i}}%
\max_{a_{i}}g_{i}(a)=0$ for $(P^{-e^{i},\eta })$.

Since $k(-e^{i})$ is bounded from above by $0$ by Lemma \ref{bound} and
Lemma \ref{minmax}, $\underline{a}^{i}$ solves $(P^{-e^{i},\eta })$. Hence $%
k^{\eta }(-e^{i})=- g(\underline{a}^{i})= -
\min_{a_{-i}}\max_{a_{i}}g_{i}(a)=0$.
\end{proof}

\bigskip

Now we complete the proof of Theorem \ref{folk}. The last two lemmas prove $%
k^{\eta }(\lambda )=\max_{a}\sum_{i}\lambda _{i}g_{i}(a)$ for any $\lambda
\notin \left\{ -e^{i},i\in N\right\} $ and $k^{\eta }(-e_{i})=0$ for every $%
i\in N$. Since $V^{\ast }(G)$ is a compact and convex set, $V^{\ast
}(G)=\bigcap_{\lambda }H^{\eta }(\lambda )=Q^{\eta }$. Then the theorem
follows from Theorem \ref{limit} when $intV^{\ast }(G)\neq \emptyset $.

\section{Discussion}

\textbf{More Strict Incentive Constraints}

\vspace{2mm}

For our uniformly strict folk theorem, we require a fixed level of strict
incentive compatibility at every public history. In terms of average payoff,
the strict incentive $(1-\delta) \eta$ converges to $0$ as $\delta
\rightarrow 1$. This means that the loss from a single deviation becomes
negligible relative to the size of the total payoff in the limit. We could
instead require $\eta$-strict incentive compatibility in terms of average
payoff. This means that the loss from a deviation is comparable to a
permanent payoff shock, say, losing \$1 in all the future periods. To do
this, we would replace $\eta$ in the definition of $\eta$-USPPE (Definition
1) with $\frac{\eta}{1-\delta}$. However, it turns out that the set of $\eta$%
-USPPE in this sense becomes empty for any $\eta > 0$ as $\delta \rightarrow
1$.

More generally, we can impose $f(\delta)$-strict incentive constraint in
terms of average payoff, where $f(\delta)$ may not converge to $0$ or
converge to $0$ more slowly than $(1-\delta)$ as $\delta \rightarrow 1$. We
can show that, for any such $f$, the set of ``$f$-uniformly strict'' PPE
becomes empty for large enough $\delta$. In this sense, our folk theorem
cannot be improved in terms of the order of the strict incentive in the
limit.

The reason for this is as follows. The effect of the current stage game payoff vanishes at the rate of $(1-\delta)$ as $\delta \rightarrow 1$ in terms of average payoff. So, if we like to provide $f(\delta)$-strict incentive with $f(\delta)$ such that $\lim_{\delta \rightarrow 1} \frac{f(\delta)}{1-\delta} \rightarrow \infty$, it must come from the variation in continuation payoffs.\footnote{For example, it is easy to see that a repetition of any strict Nash equilibrium in the stage game is not $f(\delta)$-USPPE if $\delta$ is large enough.} However, the maximum variation of continuation payoffs for player $i$ must vanish at the same rate of $(1-\delta)$ if her continuation payoff $w_i(m)$ is always at least as large as the equilibrium payoff $v$ from the present period. This is because the distance between the expected continuation payoff and the equilibrium payoff is $E[w(\cdot)|(a,\rho)] - v =\frac{1-\delta}{\delta} (v -g(a))$, which shrinks to $0$ at the rate of $1-\delta$. Hence, to provide $f(\delta)$-strict incentive, continuation payoff must be strictly less than the equilibrium payoff after some message profile, i.e., there exists $\epsilon >0$ and $m \in M$ such that $w(m) < v-\epsilon$ for any large $\delta$. However this cannot happen at every public history, hence there is no $f(\delta)$-USPPE with such $f(\delta)$ for any large enough discount factor.

\vspace{3mm}

\textbf{Folk Theorem with Double Limits}

\vspace{2mm}

We prove our folk theorem by fixing a level of strict incentive $\eta > 0$
and letting $\delta \rightarrow 1$. If we instead allow $\eta$ go to $0$ and 
$\delta$ go to $1$, then we can prove a folk theorem with weaker conditions.
When $\eta$ is small, we can construct $Q^\eta$ using the minmax action
profiles if $(G, p)$ just satisfies $\eta$-detectability instead of $\eta^*$%
- detectability at the minmax action profiles. Since $Q^\eta$ converges to $%
V^*(G)$ as $\eta \rightarrow 0$ and is the limit $\eta$-USPPE payoff set
(with full dimensionality), we can prove a version of folk theorem only with 
$\eta$-detectability (for minmax action profiles in addition to $A(G)$) and $%
\eta$-identifiability, where $\eta$ goes to $0$ and $\delta$ goes to $1$ at
the same time.

\pagebreak

\pagebreak

\section*{Appendix}

\textbf{Proof of Lemma 2}

\vspace{2mm}

\begin{proof}
This is trivial if $E^\eta(\delta)$ is an empty set, so suppose that it is
not. First, note that $E^{\eta }(\delta )$ is bounded, so we must show that $%
E^{\eta }(\delta )$ is closed. Take any $v^{\ast }\in cl(E^{\eta }(\delta )).
$ Choose a sequence $v^{k}\in E^{\eta }(\delta )$ in $\mathbb{R}^{n}$ that
converges to $v^{\ast }.$ For each $k$, let $(a^{k},\rho ^{k})\in A \times R$
be the strategy profile in the first period and $w^{k}:M\rightarrow \mathbb{R%
}^{n}$ be the continuation payoff profile from the second period of the
equilibrium strategy that supports $v^k$. Note that $w^{k}(m)\in E^{\eta
}(\delta )$ for all $m \in M$. Then for each $i\in N,$ 
\begin{equation*}
v^{k}=(1-\delta )g_{i}\left( a^{k}\right) +\delta \sum_{m\in
M}w_{i}^{k}\left( m\right) \tilde{p}\left( m|a^{k},\rho ^{k}\right)
\end{equation*}%
Since $A\times R$ is compact and $E^{\eta }(\delta )$ is bounded, we may,
extracting a subsequence if necessary, assume that $(a^{k},\rho ^{k})$ and $%
w^{k}$ are convergent with respective limits $(a^{\ast },\rho ^{\ast })$ and 
$w^{\ast }.$ Furthermore, we may assume that $(a^{k},\rho ^{k})=(a^{\ast
},\rho ^{\ast })$ for all sufficiently large $k$. Then 
\begin{equation*}
v^{\ast }=(1-\delta )g_{i}\left( a^{\ast }\right) +\delta \sum_{m\in
M}w_{i}^{\ast }\left( m\right) \tilde{p}\left( m|a^{\ast },\rho ^{\ast
}\right) .
\end{equation*}%
and for all sufficiently large k, 
\begin{align*}
& g_{i}\left( a^{\ast }\right) +\frac{\delta }{1-\delta }\sum_{m\in
M}w_{i}^{k}\left( m\right) \tilde{p}\left( m|a^{\ast },\rho ^{\ast }\right)
-\eta \geq \\
& g_{i}\left( a_{i}^{\prime },a_{-i}^{\ast }\right) +\frac{\delta }{1-\delta 
}\sum_{m\in M}w_{i}^{k}(m)\tilde{p}(m|(a_{i}^{\prime },\rho _{i}^{\prime
}),(a_{-i}^{\ast },\rho _{-i}^{\ast }))
\end{align*}%
for all $(a_{i}^{\prime },\rho _{i}^{\prime })\in \hat{\Sigma}_{i}^{(a^{\ast
},\rho ^{\ast }),w^{k}}.$ If $(a_{i}^{\prime },\rho _{i}^{\prime })\in \hat{%
\Sigma}_{i}^{(a^{\ast },\rho ^{\ast }),w^{\ast }},$ then $(a_{i}^{\prime
},\rho _{i}^{\prime })\in \hat{\Sigma}_{i}^{(a^{\ast },\rho ^{\ast }),w^{k}}$
for all sufficiently large $k$, hence in the limit 
\begin{align*}
& g_{i}\left( a^{\ast }\right) +\frac{\delta }{1-\delta }\sum_{m\in
M}w_{i}^{\ast }\left( m\right) \tilde{p}\left( m|a^{\ast },\rho ^{\ast
}\right) -\eta \geq \\
& g_{i}\left( a_{i}^{\prime },a_{-i}^{\ast }\right) +\frac{\delta }{1-\delta 
}\sum_{m\in M}w_{i}^{\ast }(m)\tilde{p}(m|(a_{i}^{\prime },\rho _{i}^{\prime
}),(a_{-i}^{\ast },\rho _{-i}^{\ast }))
\end{align*}%
for all $(a_{i}^{\prime },\rho _{i}^{\prime })\in \hat{\Sigma}_{i}^{(a^{\ast
},\rho ^{\ast }),w^{\ast }}.$ Since $w^{\ast }(m)\in cl(E^{\eta }(\delta ))$
for all $m \in M$, it follows that $v^{\ast }\in B(\delta , cl(E^{\eta
}(\delta )), \eta )$, therefore $cl(E^{\eta }(\delta ))\subseteq B(\delta
,cl(E^{\eta }(\delta )),\eta ).$ Since $cl(E^{\eta }(\delta ))$ is bounded
(in fact compact), by $\eta$-self decomposability (Lemma\ref{rec}), we can
conclude that $cl(E^{\eta }(\delta ))\subseteq E^{\eta }(\delta ),$ i.e., $%
E^{\eta }(\delta )$ is closed.
\end{proof}

\end{document}